\documentclass[10pt,letter]{IEEEtran}
\usepackage[utf8]{inputenc}          
\usepackage[T1]{fontenc}             
\usepackage[frenchb,english]{babel}  
\usepackage[noadjust]{cite}
\usepackage{graphicx}
\usepackage{float}
\usepackage{amssymb}
\usepackage[cmex10]{amsmath}
\usepackage{xfrac}
\usepackage{url}
\usepackage{authblk}
\usepackage{array,multirow, makecell}
\usepackage{algorithm}
\usepackage[noend]{algpseudocode}
\usepackage{mathtools,xparse}
\setcellgapes{4pt}
\usepackage{caption}
\usepackage{subcaption}

\def\BibTeX{{\rm B\kern-.05em{\sc i\kern-.025em b}\kern-.08em T\kern-.1667em\lower.7ex\hbox{E}\kern-.125emX}}
\makegapedcells
\usepackage[moderate]{savetrees}
\usepackage{lipsum}
\usepackage{mathtools}
\usepackage{cuted}

\newtheorem{prop}{Proposition}
\usepackage{eqnarray,amsmath}

\begin{document}


\title{On the Theoretical Limits of Beam Steering \\in mmWave Massive MIMO Channels}


\author{Mohamed Shehata}
\author{Matthieu Crussière}
\author{Maryline Hélard}
\affil{Univ Rennes, INSA Rennes, CNRS, IETR-UMR 6164, F-35000 Rennes \\email:\{mohamed.shehata, mattieu.crussiere, maryline.helard\}@insa-rennes.fr}
\renewcommand\Authands{ and }

\maketitle
\thispagestyle{empty}


\begin{abstract}
Analog Beamsteering (ABS) has emerged as a low complexity, power efficient solution for MillimeterWave (mmWave) massive Multiple Input Multiple Output (MIMO) systems. Moreover, driven by the low spatial correlation between the User Terminals (UTs) with high number of transmit antennas (massive MIMO) at the Base Station (BS), ABS can be used to support Multi User (MU) MIMO scenarios instead of digital or Hybrid Beamforming (HBF). 
However, we show in this paper, that even with high number of transmit antennas, the HBF can achieve better Spectral Efficiency (SE) compared to the MU ABS even in pure Line of Sight (LoS) channels. Moreover, we prove that the MU ABS saturates to a constant SE at high transmit Signal to Noise Ratio (SNR) and we theoretically derive  an approximation to that saturation bound in this paper. On the other hand, we highlight that the HBF's SE  scales with the transmit SNR even in high SNR regime. Finally, given the same power consumption and hardware complexity as the MU ABS case, we show that HBF asymptotically  achieves the optimal SE (ideal non interference scenario) when increasing the number of transmit antennas.
  
\end{abstract}

\begin {IEEEkeywords}
Beam Steering, Millimeter Wave (mmWave), Analog Beamforming, Analytic Analysis.
\end{IEEEkeywords}

\section{Introduction}
\label{intro}


Massive MIMO MillimeterWave (mmWave) systems have recently emerged as the main key player in the future wireless networks \cite{Boccardi_14}. Motivated by the high array gains and splendid spectrum available, massive MIMO - mmWave systems can potentially achieve the performance requirements of the next generations of cellular networks. However, achieving such gains is not forward as it is limited by a variety of challenges and bottlenecks, of which the most significant are the hardware complexity \cite{zhou_18}, power consumption \cite{Mendez_16} and the sparse channel \cite{Rangan_14}. 


Therefore, recent work in the literature \cite{Zou_16, Li_17} started to consider using analog beamforming to relax the hardware and power consumption constraints. Moreover, since mmWave channel is sparse and highly Line of Sight (LoS) dominated, Analog Beamsteering (ABS) \cite{Fettweis_16} started to emerge as one of the most attractive analog beamforming techniques driven by the fact that it only needs the LoS angular information. Henceforth, ABS can be viewed as a low complexity, low overhead solution, since it needs only estimating the LoS channel information, which are even frequency flat and do not significantly change over the channel sub-carriers\cite{Shehata_18_1}. 

However, the main limitation of the analog beamforming in general and thus of the ABS in particular, is that it can't support Multi (MU) - MIMO scenarios. Traditionaly in massive MIMO, in order to support MU MIMO scenarios digital beamforming is needed, where each transmit antenna is supplied with a Radio Frequency (RF) chain and a pair of Digital to Analog Converters (DACs). However, applying digital beamforming in mmWave massive MIMO systems requires high hardware complexity and high power consumption to satisfy its hardware requirements. Therefore, Hybrid Beamforming (HBF) \cite{Ayach_2014, Alkhateeb_2014,Alkhateeb_2015} emerged as a trade-off solution between the analog and digital beamforming scenarios. HBF utilizes a small number of RF chains and DAC pairs compared to the number of transmit antennas. Henceforth, leveraging the high transmit antenna gain and requires low hardware complexity and power consumption and can support MU MIMO scenarios. Moreover, HBF was shown to achieve SE very close to the one achieved by digital beamforming in mmWave channels \cite{Alkhateeb_2014}. This is due to the fact that at mmWave systems, the channel is sparse and thus a few number of RF chains is enough to have full access on the channel dominant paths. Moreover, it was shown recently, that HBF can achieve the same performance of digital beamforming with less hardware complexity and power consumption \cite{Molisch_17,Zhang_2005,Bogale_2016}.


Motivated by the fact that mmWave systems can employ massive MIMO systems with reasonable form factor, recent research about the favourable channel scenario and channel hardening \cite{Pawel_2019, Roy_2018, Wu_17, Ngo_2014} emerged again as practical approximations that can be asymptotically achieved. Moreover, with the emergence of these asymptotic approximations, MU ABS can be seen as a promising candidate as HBF thanks to the low interference with high number of transmit antennas. However, in this paper we show that HBF always achieves higher SE than the MU ABS. Moreover, although it was shown in the literature \cite{Alkhateeb_2015} that MU ABS saturates in SE at high SNR, no closed form model for this saturation bound has been offered until now due to its mathematical complexity. In this paper, we provide analytic closed form model for this saturation bound for MU ABS in high SNR regime for mmWave LoS channels. Moreover, we validate our model using simulation results. 

\section{System and Channel Model}
\label{section3}

The system introduced thrughout the paper is a downlink narrow band MU Multiple Input Single Output (MISO) system. The Base Station (BS) has $N_t$ transmit antennas organized in a Uniform Linear Array (ULA) architectures, and serving $K$ User Terminals (UTs) each equipped with a single receive antenna.

The BS applies beamforming to serve the $K$ UTs, such that the received signal vectors $\mathbf{r}=[R_1, R_2,...,R_{K}]^T \in \mathbb{C}^{K \times 1}$ can be calculated as follows: 
\begin{equation}
\mathbf{r} = \mathbf{H}\mathbf{x} + \mathbf{n}
\end{equation}

such that  $\mathbf{n}=[N_1,N_2,...,N_{K}]^T \in \mathbb{C}^{K \times 1}$ denotes the independent and identically distributed (i.i.d.) Additive White Gaussian Noise (AWGN) vector where $n \sim \mathcal{N}(0,\,\sigma_{n}^{2})$, and $\sigma_{n}^{2}$ denotes the noise variance. $\mathbf{H} \in \mathbb{C}^{K \times N_{t}}$ represents the MU MISO propagation channel, while the beam-formed transmit symbols vector is denoted as $\mathbf{x}=[X_1, X_2,...,X_{N_{t}}]^T \in \mathbb{C}^{N_{t} \times 1}$ and can be expanded as:  

\begin{equation}
\mathbf{x} = \mathbf{F} \mathbf{s}
\label{Eq_precode}
\end{equation}

such that the beamforming matrix is represented as $\mathbf{F}=[\mathbf{f}_1, \mathbf{f}_2 , ..., \mathbf{f}_{K}]^T$ and will be explained in further details in the next section, while $\mathbf{s}=[S_1, S_2,...,S_{K}]^T \in \mathbb{C}^{ K \times 1}$ denotes the transmitted symbols vector before beamforming. 


The channel model used throughout the paper is the sparse geometric channel model \cite{Sayed_2002,Alkhateeb_2015}, which is used in most of the literature of mmWave signal processing techniques \cite{Alkhateeb_2014, Ayach_2014, Alkhateeb_2015}. This model, also known as, ray-based channel model, describes the channel by the paths (physical rays) that exist between the transmitter and the receiver.

Henceforth, the channel vector $\mathbf{h}_k$ between each UT $k$ and the BS can be expressed as:
\begin{equation}
\mathbf{h}_k = \sqrt{\frac{N_{t}}{P_k}}\sum_{p=1}^{P_k} \alpha_{k,p} \mathbf{a}_{t}^H(\phi_{k,p}) \,
\label{channel}
\end{equation}

such that $\alpha_{k,p}$ represents the $p^{th}$ propagation path complex amplitude for UT $k$, where $P_k$ denotes the total number of paths that can be received by UT $k$ and $\alpha \sim \mathcal{CN}(0,2 \sigma^2)$, here we assume $2 \sigma^2 = 1$. The Angle of Departure (AoD) for each path $p$ for UT $k$ is denoted as $\phi_{k,p}$ and assumed to be uniformly distributed $\phi \sim \mathcal{U}[0,2\pi] $. 

 The transmit array steering vector is represented as $\mathbf{a}_{t}(\phi_{k,p})$ , given that the BS deploys a ULA array, $\mathbf{a}_{t}(\phi_{k,p})$ can be defined as:
\begin{equation}
\mathbf{a}_{t}(\phi_{k,p})= \frac{1}{\sqrt{N_{t}}} [1, e^{j\zeta(\phi_{k,p})}, ...,e^{j(N_{t}-1)\zeta(\phi_{k,p})}]^T
\end{equation}

where $\zeta(\phi_{k,p})$ is defined as:
\begin{equation}
\zeta(\phi_{k,p}) = \frac{2\pi}{\lambda}d\sin(\phi_{k,p})
\end{equation}

such that $d$ represents the inter-element antenna spacing, while $\lambda$ represents the wavelength of the signal. Then, the MU MISO channel matrix $\mathbf{H} \in \mathbb{C}^{K \times N_{t}}$ can be represented as 

\begin{equation}
\mathbf{H}= [\mathbf{h}_{1}^T,\mathbf{h}_{2}^T, ...,\mathbf{h}_{K}^T]^T
\end{equation}

Although this sparse geometric channel model is favourable for mmWave practical channel statistical modelling, it is not favourable in terms of statistical analysis and closed form performance modelling compared to the classical statistical channel models (i.i.d Rayleigh, correlated Rayleigh, Ricean, ..etc). Up to our knowledge only a few papers have been involved on the statistical signal processing analysis of such channels \cite{Roy_2018, Pawel_2019,Pawel_2019_2 } which leaves a gap in the current literature of mmWave MIMO signal processing. 

In this paper we try to unwrap some of the mathematical features of the pure LoS case of this model and provide closed form theoretical bounds for the analog and Hybrid Beamsteering (HBS) approaches. We consider a pure LoS scenario ($P_k=1, \forall K $), which is a fairly acceptable assumption for mmWave channels at high frequencies \cite{Marco_2016}. Therefore, substituting $P_k=1$ in Equation (\ref{channel}), the pure LoS channel for each UT $k$ can be expressed as:
\begin{equation}
\mathbf{h}_k = \sqrt{{N_{t}}} ~ \alpha_{k} \mathbf{a}_{t}^H(\phi_{k})   
\label{channel_2}
\end{equation}

\section{Beamforming Strategies}
\label{section4}

In this section, we describe in details the ABS and HBF adopted throughout the paper.

\subsection{Analog Beamsteering}

In this analog beamforming scenario, the phase shifters in the analog domain are adjusted in order to steer the beam for the UT $k$ over the LoS path. Henceforth, analog beamsteering has attracted the attention recently in mmWave channels, since they are LoS dominated \cite{Alkhateeb_2015}. Moreover, analog beamsteering has low overhead requirements, since it only requires estimating the LoS channel, which is flat in frequency response \cite{Shehata_18_1}. Given that we adopt the pure LoS channel in Equation (\ref{channel_2}), the analog beamformer $\mathbf{f}_{RF,k}$ for UT $k$ in this case can be expressed as:

\begin{equation}
    \mathbf{f}_{RF,k}= \mathbf{a}_{t}(\phi_{k})
    \label{ABS}
\end{equation}

\subsection{Hybrid Beamforming}

In order to extend the aforementioned ABS to consider MU scenarios with taking into account the Inter User Interference (IUI) between the UTs, we propose in this subsection the Hybrid-BeamSteering (HBS) architecture, which is a hybrid beamforming evolution of the ABS by adding a digital Zero Forcing (ZF) precoding layer in the BaseBand (BB) to mitigate the IUI. Throughout this paper we consider HBS as the only HBF architecture, therefore both HBF and HBS notations are similar in this paper. This HBS is familiar in the mmWave MIMO literature \cite{Alkhateeb_2015, Shehata_2018_2} and it aims at decoupling the MU beamforming matrix $\mathbf{F}$ in Equation (\ref{Eq_precode}) into two parts namely ABS in the analog (RF) part and ZF in the digital (BB) part. Henceforth, the HBS beamforming matrix $\mathbf{F}_{HBS}$ can be expressed as:

\begin{equation}
    \mathbf{F}_{HBS}= \mathbf{F}_{RF} \mathbf{W}_{ZF} 
    \label{HBS}
\end{equation}
 
where $\mathbf{F}_{RF}= [\mathbf{f}_{RF,1} ,..., \mathbf{f}_{RF,K}], \mathbf{F}_{RF}  \in \mathbb{C}^{N_{t} \times N_{RF}} $ denotes ABS matrix, with the column vectors given in Equation (\ref{ABS}), where $N_{RF}$ is the number of RF chains at the transmitter. In this paper we use number of RF chains equal to the number of UTs $N_{RF}=K$, since each UT is served by a single stream. $\mathbf{W}_{ZF} \in \mathbb{C}^{N_{RF} \times K} $ is the ZF digital precoding matrix. In order to calculate the $\mathbf{W}_{ZF}$, we first calculate the equivalent channel vector for each UT $k$ $\hat{\mathbf{h}}_{k}\in \mathbb{C}^{1 \times N_{RF}}$ as $\hat{\mathbf{h}}_{k}= {\mathbf{h}}_{k}\mathbf{F}_{RF}$.

Then, the total MU equivalent channel for the $K$ UTs $\hat{\mathbf{H}}$ can be expressed as:
\begin{equation}
\hat{\mathbf{H}}=   [\hat{\mathbf{h}}_{1}^T ,...,\hat{\mathbf{h}}_{K}^T] 
\label{Equivalent_H}
\end{equation}

where $\hat{\mathbf{H}}$ is the channel seen at the digital layer, hence the digital precoding matrix $\mathbf{W}_{ZF}$ can be calculated as:

\begin{equation}
\mathbf{W}_{ZF} = \hat{\mathbf{H}}^{H}(\hat{\mathbf{H}} \hat{\mathbf{H}}^{H})^{-1} 
\end{equation}

Then, $\mathbf{W}_{ZF}$ is normalized to satisfy the total power constraint. In this paper, we use the Vector Normalization (VN) method, since it is shown in the literature that it outperforms the Matrix Normalization (MN) method in terms of SE \cite{Lim_2015}. Therefore applying the VN method on the digital precoding matrix column as follows:

\begin{equation}
   \mathbf{w}_{ZF,k} = \frac{\mathbf{w}_k}{{\lVert \mathbf{f}_{k}^{HBF}  \rVert}} =\frac{\mathbf{w}_k}{{\lVert \mathbf{F}_{RF}\mathbf{w}_{k} \rVert }} 
\end{equation}

Reconstructing $\mathbf{W}_{ZF}$ again as $\mathbf{W}_{ZF}= [\mathbf{w}_{ZF,1} ,..., \mathbf{w}_{ZF,K}]$ and recalling Equation (\ref{HBS}), the HBS beamforming matrix $\mathbf{F}_{HBS}$ can be calculated.

\section{Spectral Efficiency Analysis}

In this section, we provide analytical analysis for the theoretical bounds of the achievable SE for MU ABS and MU HBS.
\subsection{SE Analysis for ABS}
In this subsection, we will provide SE analysis for the MU ABS, where in this case multiple UTs are served simultaneously in the same time-frequency resource using ABS and the IUI is not tackled. The expectation of the per stream SE $\eta_k$ in this case can be approximated as follows:

\begin{equation}
    \mathbb{E}[\eta_k] = \mathbb{E} \left \{ \log_2 \Bigg( 1 + \frac{\rho \vert \mathbf{h}_{k} \mathbf{f}_{RF,k} \vert ^2}{\rho \sum_{i=1, i \neq k}^{K} \vert \mathbf{h}_{k} \mathbf{f}_{RF,i} \vert ^2 + 1} \Bigg) \right  \}
    \label{SE_exp}
\end{equation}

where $\rho$ represents the per UT transmit SNR. Utilizing Equation (\ref{SE_exp}), we will define two propositions that characterize the SE performance of MU ABS at high transmit SNR regime.

\begin{prop}
For $K=2$ UTs, served by MU ABS in pure LoS channel, the expected achieved per stream SE $\mathbb{E} [\eta_k]$ saturates at high SNR regime to a constant value approximated by:
\begin{equation}
  \mathbb{E} [\eta_k] \approx \log_2\Bigg( 1 + \frac{N_t^2}{ \Big(1 + 2 \sum_{i=1}^{N_t - 1} \Big(1 - \frac{i}{N_t}\Big)\mathcal{J}_{0}^{2}(2 \pi d i) \Big) }\Bigg)
   \label{lemma_1}
\end{equation}

where $\mathcal{J}_{0}$ represents the zero order Bessel function.
\end{prop}

\begin{proof}
Here we use the approximation that $\mathbb{E}\Big[\log_2 \Big(1 + \frac{\mathbf{X}}{\mathbf{Y}}\Big)\Big] \approx \log_2 \Big(1 + \frac{\mathbb{E}[\mathbf{X}]}{\mathbb{E}[\mathbf{Y}]}\Big) $ in \cite{Matthaiou_2014} if $\mathbf{X}= \sum \mathbf{X}_i$ and $\mathbf{Y}= \sum \mathbf{Y}_i$ represent summation of non negative random variables. $\mathbf{X}$ and $\mathbf{Y}$ in-dependency is not required for this approximation to hold, and the approximation accuracy increases with the number of the summation terms included in $\mathbf{X}$ and $\mathbf{Y}$ \cite{Pawel_4, Matthaiou_2014}. Therefore, Equation (\ref{SE_exp}) can now be approximated (given that $K=2$ UTs) as follows:

\begin{equation}
    \mathbb{E}[\eta_k] \approx  \log_2 \Bigg( 1 + \frac{\mathbb{E}[\rho \vert \mathbf{h}_{k} \mathbf{f}_{RF,k} \vert ^2]}{\mathbb{E}[\rho \vert \mathbf{h}_{k} \mathbf{f}_{RF,\hat{k}} \vert ^2+ 1]} \Bigg)
    \label{SE_exp_2}
\end{equation}

where $\mathbf{f}_{RF,\hat{k}}$ is the ABS beamformer of the interfering UT $\hat{k}$. At high SNR regime ($\rho \to \infty$), the interference dominates the noise and thus it can be deduced that $\rho \vert \mathbf{h}_{k} \mathbf{f}_{RF,\hat{k}} \vert ^2 >> 1$. Therefore, at high SNR the approximation $\mathbb{E}[\rho \vert \mathbf{h}_{k} \mathbf{f}_{RF,\hat{k}} \vert ^2+ 1] \approx \mathbb{E}[\rho \vert \mathbf{h}_{k} \mathbf{f}_{RF,\hat{k}} \vert ^2]$ holds. Therefore, Equation (\ref{SE_exp_2}) can be reformulated as:

\begin{equation}
    \mathbb{E}[\eta_k] \approx \log_2 \Bigg( 1 + \frac{\mathbb{E}[ \vert \mathbf{h}_{k} \mathbf{f}_{RF,k} \vert ^2]}{\mathbb{E}[ \vert \mathbf{h}_{k} \mathbf{f}_{RF,\hat{k}} \vert ^2]} \Bigg)
    \label{SE_exp_3}
\end{equation}

Henceforth, it is clear that at high SNR regime $\mathbb{E}[\eta_k]$ does not depend on $\rho$ anymore and saturates to a constant value. In order to evaluate Equation (\ref{SE_exp_3}), we evaluate $\mathbb{E}[ \vert \mathbf{h}_{k} \mathbf{f}_{RF,k} \vert ^2]$ and $\mathbb{E}[\vert \mathbf{h}_{k} \mathbf{f}_{RF,\Hat{k}} \vert ^2 ]$ separately as follows:

\begin{equation}
\begin{aligned}
    \mathbb{E}[ \vert \mathbf{h}_{k} \mathbf{f}_{RF,k} \vert ^2] &=\mathbb{E}[ \vert  \sqrt{{N_{t}}} ~ \alpha_{k} \mathbf{a}_{t}^H(\phi_{k})\mathbf{a}_{t}(\phi_{k}) \vert ^2] \\
    &= N_t \mathbb{E}[\vert \alpha_{k} \vert ^2] = N_t
\end{aligned}
\label{sig_final}
\end{equation}

where $\vert \alpha_{k} \vert ^2$ has a chi-squared distribution. Then evaluating the interference part $\mathbb{E}[\vert \mathbf{h}_{k} \mathbf{f}_{RF,\Hat{k}} \vert ^2 ]$ as follows:

\begin{equation}
\begin{aligned}
\mathbb{E}[\vert \mathbf{h}_{k} \mathbf{f}_{RF,\Hat{k}} \vert ^2 ] &=\mathbb{E}[ \vert  \sqrt{{N_{t}}} ~ \alpha_{k} \mathbf{a}_{t}^H(\phi_{k})\mathbf{a}_{t}(\phi_{\Hat{k}}) \vert ^2] \\
 &= N_t \mathbb{E}[\vert \alpha_{k} \vert ^2] \mathbb{E}[\vert\mathbf{a}_{t}^H(\phi_{k})\mathbf{a}_{t}(\phi_{\Hat{k}}) \vert ^2]
\end{aligned}
\label{int_1}
\end{equation}

where $\alpha$ and $\phi$ are statistically independent which explains the second line in Equation (\ref{int_1}) and as aforementioned $\mathbb{E}[\vert \alpha_{k} \vert ^2]= 1$. According to \cite{Pawel_2019}, given that $\phi \sim \mathcal{U}[0,2\pi]$, $\mathbb{E}[\vert\mathbf{a}_{t}^H(\phi_{k})\mathbf{a}_{t}(\phi_{\Hat{k}}) \vert ^2]$ can be expressed as follows:

\begin{equation}
    \mathbb{E}[\vert\mathbf{a}_{t}^H(\phi_{k})\mathbf{a}_{t}(\phi_{\Hat{k}}) \vert ^2]= \frac{1 + 2 \sum_{i=1}^{N_t - 1} \Big(1 - \frac{i}{N_t}\Big)\mathcal{J}_{0}^{2}(2 \pi d i)}{N_t^2}
    \label{off_diag}
\end{equation}

Therefore, $\mathbb{E}[\vert \mathbf{h}_{k} \mathbf{f}_{RF,\Hat{k}} \vert ^2 ]$ can be expressed as:

\begin{equation}
    \mathbb{E}[\vert \mathbf{h}_{k} \mathbf{f}_{RF,\Hat{k}} \vert ^2 ] = \frac{1 + 2 \sum_{i=1}^{N_t - 1} \Big(1 - \frac{i}{N_t}\Big)\mathcal{J}_{0}^{2}(2 \pi d i)}{N_t}
    \label{int_final}
\end{equation}

substituting Equations (\ref{sig_final}) and (\ref{int_final}) in Equation (\ref{SE_exp_3}), the saturation level of the expected per stream SE $\mathbb{E} [\eta_k]$ for $K=2$ UTs using MU ABS at high SNR regime can be approximated as follows:

\begin{equation}
  \mathbb{E} [\eta_k] \approx \log_2\Bigg( 1 + \frac{N_t^2}{1 + 2 \sum_{i=1}^{N_t - 1} \Big(1 - \frac{i}{N_t}\Big)\mathcal{J}_{0}^{2}(2 \pi d i) }\Bigg)
\end{equation}

\end{proof}

\begin{prop}
For $K > 2$ UTs, served by MU ABS in pure LoS channel, the expected achieved per stream SE $\mathbb{E} [\eta_k]$ saturates at high SNR regime with large number of transmit antennas to a constant value approximated by:

\begin{equation}
   \mathbb{E} [\eta_k] \approx \log_2\Bigg( 1 + \frac{N_t^2}{(K-1)^2 \Big(1 + 2 \sum_{i=1}^{N_t - 1} \Big(1 - \frac{i}{N_t}\Big)\mathcal{J}_{0}^{2}(2 \pi d i) \Big) }\Bigg)
   \label{lemma_2}
\end{equation}

\end{prop}

\begin{proof}
The difference between this case and the previous proposition is only in the interference term, since here we have $K-1$ interference terms instead of $1$ in the previous proposition. Therefore, similar to the analysis in the previous proposition, at high SNR the average per stream SE $\mathbb{E} [\eta_k]$ can be approximated as follows:

\begin{equation}
    \mathbb{E}[\eta_k] \approx   \log_2 \Bigg( 1 + \frac{\mathbb{E}[ \vert \mathbf{h}_{k} \mathbf{f}_{RF,k} \vert ^2]}{\mathbb{E}[\sum_{i=1, i \neq k}^{K} \vert \mathbf{h}_{k} \mathbf{f}_{RF,i} \vert ^2]} \Bigg)
    \label{prop_2_1}
\end{equation}

where $\mathbb{E}[ \vert \mathbf{h}_{k} \mathbf{f}_{RF,k} \vert ^2]= N_t$ as shown in Equation (\ref{sig_final}), while $\mathbb{E}[\sum_{i=1, i \neq k}^{K} \vert \mathbf{h}_{k} \mathbf{f}_{RF,i} \vert ^2]= \mathbb{E}[ \sum_{i=1, i \neq k}^{K}  \vert  \sqrt{{N_{t}}} ~ \alpha_{k} \mathbf{a}_{t}^H(\phi_{k})\mathbf{a}_{t}(\phi_{i}) \vert ^2 ] $ can be expanded as follows:

\begin{equation}
   \mathbb{E}[ \sum_{i=1, i \neq k}^{K} \vert \mathbf{h}_{k} \mathbf{f}_{RF,i} \vert ^2 ] =  N_t \mathbb{E}[\vert \alpha_{k} \vert ^2]  \mathbb{E}[ \sum_{i=1, i \neq k}^{K} \vert\mathbf{a}_{t}^H(\phi_{k})\mathbf{a}_{t}(\phi_{i}) \vert ^2]
   \label{off_diag_1}
\end{equation}

Given the fact that $\sum_{i=1}^{n} x_{i}^2 \leq (\sum_{i=1}^{n} x_i)^2$, therefore, here we use the upper bound approximation $\sum_{i=1, i \neq k}^{K}  \vert  \mathbf{a}_{t}^H(\phi_{k})\mathbf{a}_{t}(\phi_{i}) \vert ^2 \approx (\sum_{i=1, i \neq k}^{K}  \vert  \mathbf{a}_{t}^H(\phi_{k})\mathbf{a}_{t}(\phi_{i}) \vert) ^2$ for tractability issues, given that the approximation error is minimal for high number of transmit antennas $N_t$ since the correlation terms $\vert \mathbf{a}_{t}^H(\phi_{k})\mathbf{a}_{t}(\phi_{i}) \vert$ will have values that tend to zero. Therefore, according to \cite{Aldaz_2015},  $\mathbb{E}[(\sum_{i=1, i \neq k}^{K}  \vert  \mathbf{a}_{t}^H(\phi_{k})\mathbf{a}_{t}(\phi_{i}) \vert) ^2]$ can be approximated by its upper bound as follows:

\begin{equation}
    \mathbb{E}[\Big(\sum_{i=1, i \neq k}^{K}  \vert  \mathbf{a}_{t}^H(\phi_{k})\mathbf{a}_{t}(\phi_{i}) \vert\Big) ^2] \approx (K-1) \sum_{i=1, i \neq k}^{K} \mathbb{E}[\vert  \mathbf{a}_{t}^H(\phi_{k})\mathbf{a}_{t}(\phi_{i}) \vert ^2 ]
\end{equation}

with approximation error $\Gamma$ given as follows:

\begin{equation}
\Gamma = \frac{1}{2} \sum_{i=1, i \neq k}^{K} \sum_{j=1, j \neq k}^{K} (\vert  \mathbf{a}_{t}^H(\phi_{k})\mathbf{a}_{t}(\phi_{i}) \vert - \vert  \mathbf{a}_{t}^H(\phi_{k})\mathbf{a}_{t}(\phi_{j}) \vert)^2
\end{equation}

where $\Gamma \to 0$ when the number of transmit antennas is large ($N_t \to \infty$). Therefore, similar to Equation (\ref{off_diag}) and according to \cite{Pawel_2019}, $\delta = (K-1) \sum_{i=1, i \neq k}^{K} \mathbb{E}[\vert  \mathbf{a}_{t}^H(\phi_{k})\mathbf{a}_{t}(\phi_{i}) \vert ^2 ]$ can be expressed as follows:

\begin{equation}
    \delta = (k-1)^2 \frac{1 + 2 \sum_{i=1}^{N_t - 1} \Big(1 - \frac{i}{N_t}\Big)\mathcal{J}_{0}^{2}(2 \pi d i)}{N_t^2}
    \label{off_diag_2}
\end{equation}

Finally, from Equations (\ref{off_diag_2}), (\ref{off_diag_1}) and (\ref{prop_2_1}), the saturation level of the expected per stream SE $\mathbb{E} [\eta_k]$ for $K > 2$ UTs using MU ABS at high SNR regime with large number of transmit antennas can be approximated as follows:

\begin{equation}
   \mathbb{E} [\eta_k] \approx \log_2\Bigg( 1 + \frac{N_t^2}{(K-1)^2 \Big(1 + 2 \sum_{i=1}^{N_t - 1} \Big(1 - \frac{i}{N_t}\Big)\mathcal{J}_{0}^{2}(2 \pi d i) \Big) }\Bigg)
\end{equation}

\end{proof}

\subsection{SE Analysis for HBF}

In this subsection, the tight upper bound approximation of the achievable per stream SE $\mathbb{E} [\eta_k]$ for MU HBS with large number of transmit antennas, is analysed in the following proposition which was initially proposed in \cite{Mokh_2019}. This proposition will be introduced here for the sake of clarity to the reader, and to show the privilege of HBS over ABS for MU scenarios, specifically in high SNR regime.

\begin{prop}
For $K$ UTs, served by MU HBS in pure LoS channel with large number of transmit antennas, the expected achieved per stream SE $\mathbb{E} [\eta_k]$ can be approximated by:

\begin{equation}
   \mathbb{E} [\eta_k] \approx \frac{2}{\ln{2}} \left(\ln(\sqrt{\rho N_t}\sigma) + \frac{\ln(2)}{2} - \frac{\kappa}{2}\right)
   \label{SE_model_approx}
\end{equation}
where $\kappa \approx 0.5772$ is the Euler constant and $\sigma$ represents the standard deviation of the complex Gaussian channel coefficients $\alpha$.

\end{prop}

\begin{proof}
Given that ZF is applied in the digital layer, mitigating the IUI is ensured, therefore, the per stream SE can be approximated, given that a large number of transmit antennas $N_t$ is used at the BS, as follows:

\begin{eqnarray}
\begin{aligned}
\mathbb{E} [\eta_k] & \approx \mathbb{E}\left\{ \log_2 \left(1 + \rho \vert \mathbf{h}_k \mathbf{f}_{RF,k} \vert ^2 \right) \right\} \\
& \approx \mathbb{E}\left\{ \log_2 \left( 1 + \rho N_t  \vert \alpha_k \vert^2 \right) \right\}.
\label{diag_H_eq_4}
\end{aligned}
\end{eqnarray}

For high SNR $\rho$ and high number of transmit antennas $N_t$, the assumption $\rho N_t  \vert \alpha_k \vert^2 >> 1$ holds, which leads to the approximation $ 1 + \rho N_t  \vert \alpha_k \vert^2 \approx  \rho N_t  \vert \alpha_k \vert^2 $. Therefore, $\mathbb{E} [\eta_k]$ can be represented after some simple mathematical manipulations as:

\begin{eqnarray}
\begin{aligned}
\mathbb{E} [\eta_k] \approx \frac{2}{\ln{2}} \mathbb{E}\left\{\ln{\left( \sqrt[]{\rho N_t} \vert \alpha_k \vert \right) } \right\}
\end{aligned}
\label{aprrox_SE}
\end{eqnarray}

such that $\ln{\left( \sqrt[]{\rho N_t} \vert \alpha_k \vert \right) }$ follows Log-Rayleigh distribution \cite{Rivet_2007} as follows $\ln{\left( \sqrt[]{\rho N_t } \vert \alpha_k \vert \right)} \sim \mathbf{LogRay} (\rho N_t\sigma^2)$. Therefore, $\mathbb{E}\left\{\ln{\left( \sqrt[]{\rho N_t} \vert \alpha_k \vert \right) } \right\}$ is given according to \cite{Rivet_2007} as follows:

\begin{equation}
    \mathbb{E}\left\{\ln{\left( \sqrt[]{\rho N_t} \vert \alpha_k \vert \right) } \right\} = \ln(\sqrt{\rho N_t}\sigma) + \frac{\ln(2)}{2} - \frac{\kappa}{2}
    \label{lemma_3_2}
\end{equation}

where $\kappa \approx 0.5772$ is the Euler constant as aforementioned. Finally, substituting Equation (\ref{lemma_3_2}) in (\ref{aprrox_SE}), the expected achievable per stream SE $\mathbb{E} [\eta_k]$ for $K$ UTs, served by MU HBS in pure LoS channel with large number of transmit antennas,  can be approximated by:

\begin{equation}
   \mathbb{E} [\eta_k] \approx \frac{2}{\ln{2}} \left(\ln(\sqrt{\rho N_t}\sigma) + \frac{\ln(2)}{2} - \frac{\kappa}{2}\right)
   \label{SE_model_approx_proof}
\end{equation}
\end{proof}

Henceforth, by observing Equations (\ref{lemma_1}), (\ref{lemma_2}) and (\ref{SE_model_approx}), it is clear that MU HBS is more favourable than MU ABS in high SNR regimes since it scales with SNR, while MU ABS doesn't scale at high SNR values. Moreover, the SE gap between MU HBS and MU ABS can be approximated in a closed form for a given transmit SNR thanks to our derived models.

\section{Numerical Analysis}
\label{section5}

\begin{figure}[t!]
   \centering
   \includegraphics[scale=.5]{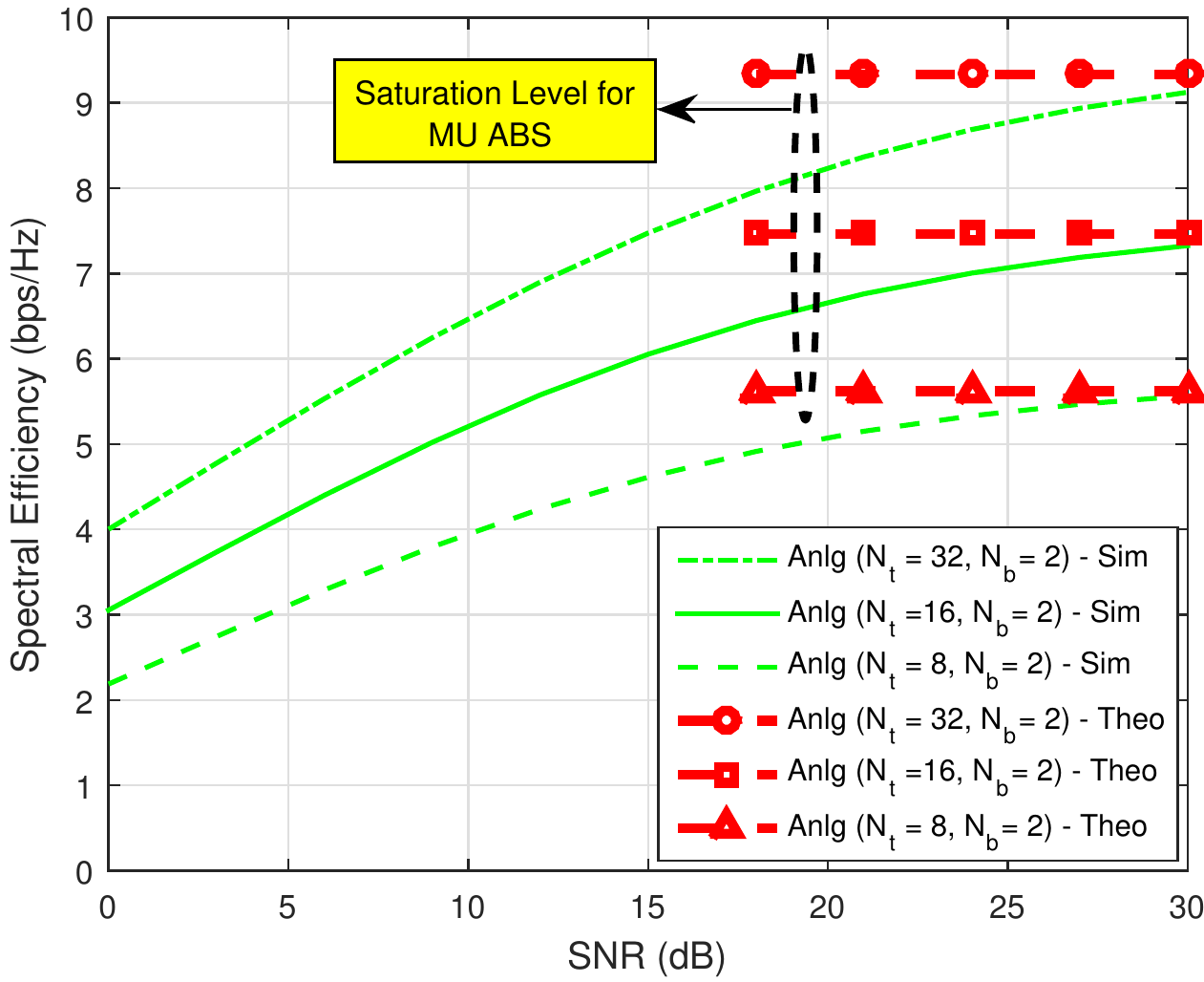}
    \caption{The simulated per stream SE for analog beam steering, together with the theoretical saturation bound for different values of $N_t$ and $N_b= 2$. }
    \label{fig_1}
\end{figure}

In this section, we validate the aforementioned SE models for both MU ABS and MU HBS in pure LoS channel. The transmit antenna array is ULA with half wavelength spacing $d= \frac{\lambda}{2}$. The number of transmit RF chains $N_{RF}$ equals the number of UTs and thus equals the number of the steered beams $N_b$ by the BS $N_{RF}= K = N_b$ for both MU ABS and MU HBS. The simulations are carried out in a Monte Carlo fashion with $50000$ realizations. Perfect Channel State Information at the Transmitter (CSIT) is assumed.

In Figure \ref{fig_1}, the simulated per stream SE for MU ABS, together with the theoretical saturation bound given in Equation (\ref{lemma_1}) are evaluated for different values of $N_t$, given $K=N_b= 2$ for validation purposes. We can observe that for different $N_t$, the SE for MU ABS saturates to its corresponding analytical saturation upper bound at high SNR which validates our SE model in Equation (\ref{lemma_1}).

In Figure \ref{fig_2}, the simulated per stream SE for MU HBS, together with the theoretical upper bound approximation given in Equation (\ref{SE_model_approx}) are evaluated for different values of $N_t$, given $K=N_b= 2$ for validation purposes. We can observe that for low SNR values the model in Equation (\ref{SE_model_approx}) and refereed to in the figure as ('No Interference - Theo') achieves a bit lower SE compared to the simulated upper bound (diagonal equivalent channel with no interference) due to the approximation $ 1 + \rho N_t  \vert \alpha_k \vert^2 \approx  \rho N_t  \vert \alpha_k \vert^2 $. However, for intermediate and high transmit SNR values, the approximation is tight. Moreover, we can observe that for high number of transmit antennas $N_t=128$ and $N_b =2$, the HBS achieves approximately the same SE performance as Equation (\ref{SE_model_approx}), while for $N_t= 32, N_b=2$ the approximation error is about $\approx 0.2$ b/s/Hz at $\rho = 30$ dB and for  $N_t= 16, N_b=2$ the approximation error is about $\approx 0.3$ b/s/Hz at $\rho = 30$ dB. Therefore, validating Equation (\ref{SE_model_approx}) as a tight upper bound for the per stream SE in case $K=N_b=2$. Moreover, observing Figures \ref{fig_1} and \ref{fig_2} together, we can infer that for high SNR regimes, MU HBS is favourable compared to MU ABS, since it scales with transmit SNR and also we can quantify the SE performance gap between both for a given transmit SNR value.

	\begin{figure}[t!]
   \centering
   \includegraphics[scale=.5]{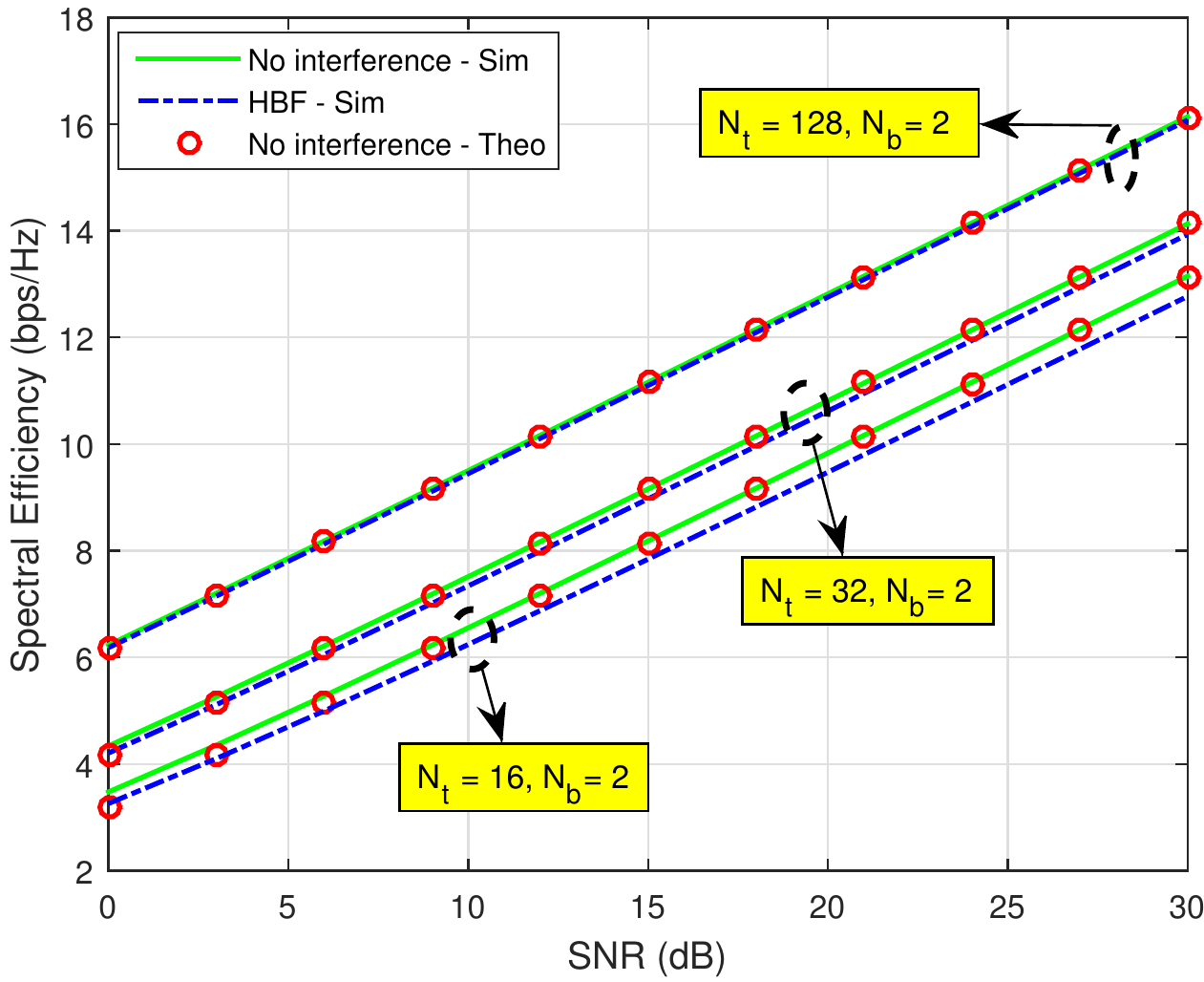}
    \caption{The simulated per stream SE for HBF, together with the theoretical upper bound approximation for different values of $N_t$ and $N_b= 2$.}
    \label{fig_2}
\end{figure}

Moving to Figure \ref{fig_3}, the simulated per stream SE for MU ABS and HBS, together with the theoretical bounds in Equations (\ref{lemma_2}) and (\ref{SE_model_approx}) are evaluated for different values of $N_b$ and given $N_t= 32$ for validation purposes. It is shown that for an average number of transmit antennas $(N_t=32)$, the theoretical saturation bound in Equation (\ref{lemma_2}) is able achieve an accurate approximation for the simulated saturation bound with an SE approximation error of $0.1$ b/s/Hz for $N_b= 3, \rho= 30$ dB  and $0.15$ b/s/Hz for $N_b= 5, \rho= 30$ dB, which validates the tightness of our approximation in Equation (\ref{lemma_2}). However, it is shown that at an average number of transmit antennas $(N_t=32)$, and a large number of UTs ($K=N_b=5$), the tightness of the approximation in Equation (\ref{SE_model_approx}) is no longer guaranteed and it acts as an upper bound more than an approximation. In this case ($N_t=32,N_b=5, \rho=30$ dB) the SE approximation error between the simulated HBS and the approximation in Equation (\ref{SE_model_approx}) is $\approx 1$ b/s/Hz. Therefore, it can be concluded that for large number of UTs $K$, for the approximation in Equation (\ref{SE_model_approx}) to be considered as tight approximation, the number of transmit antennas should be high enough ($N_t >> K$). 

In Figure \ref{fig_4}, the simulated per stream SE for MU ABS and HBS, together with the theoretical bounds in Equations (\ref{lemma_2}) and (\ref{SE_model_approx}) are evaluated for $K=N_b=5$ and given $N_t= 128$ for validation purposes. As aforementioned, we can observe that increasing the number of transmit antennas such that $N_t >> K$ enhcanced both approximations in Equations (\ref{lemma_2}) and (\ref{SE_model_approx}) and specifically the MU HBS one in Equation (\ref{SE_model_approx}) as expected. We can observe that for ($N_t=128,N_b=5, \rho=30$ dB) the SE approximation error between the simulated HBS and the approximation in Equation (\ref{SE_model_approx}) is $\approx 0.2$ b/s/Hz compared to $\approx 1$ b/s/Hz in the case $N_t = 32$ in Figure \ref{fig_3}. Moreover, we can observe that for the same scenario ($N_t=128,N_b=5, \rho=30$ dB), the SE approximation error between the simulated saturation level of ABS and the approximation in Equation (\ref{lemma_2}) is $0.1$ b/s/Hz compared to $0.15$ b/s/Hz in the case $N_t = 32$ in Figure \ref{fig_3}.

\begin{figure}[t!]
   \centering
   \includegraphics[scale=.5]{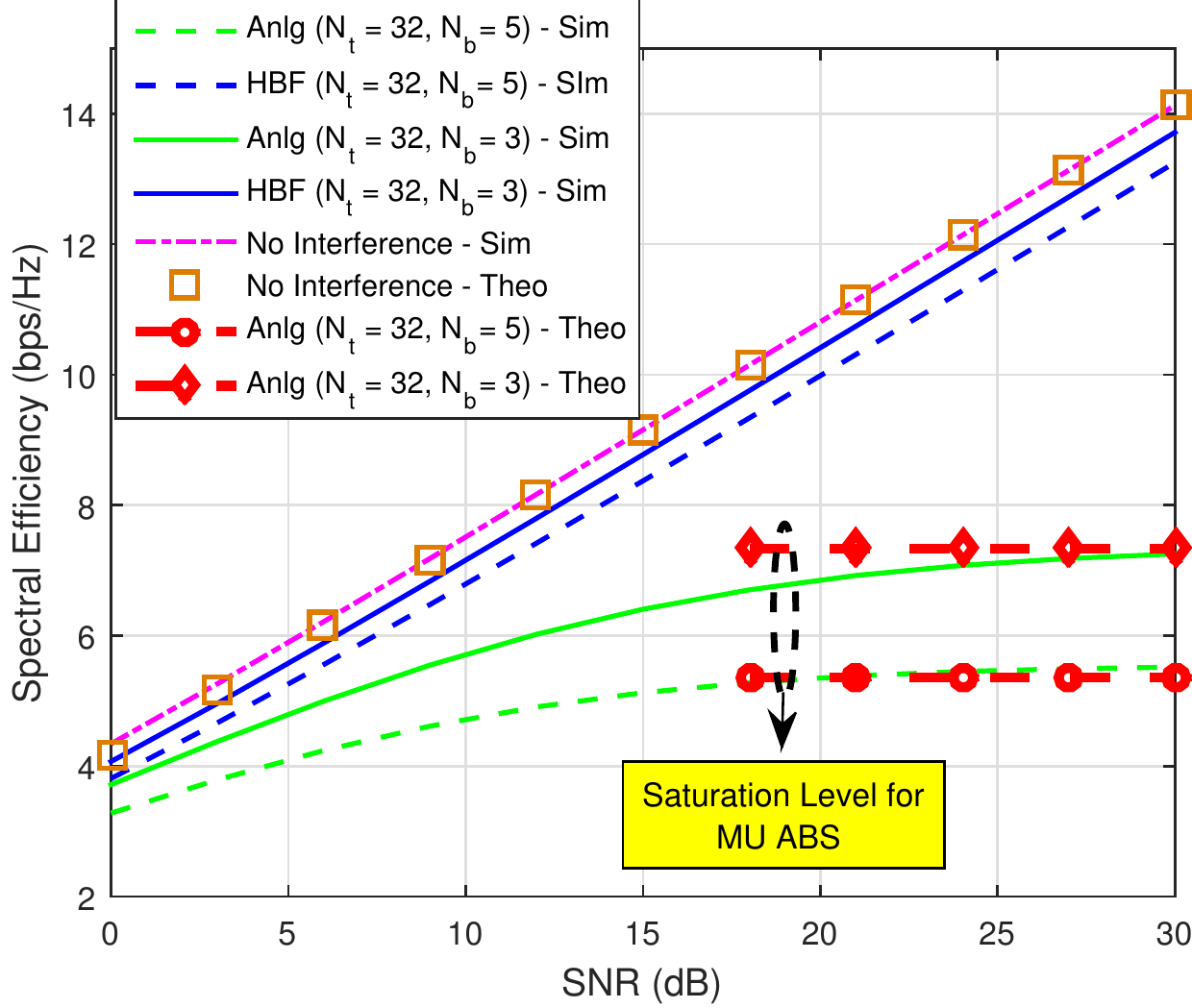}
    \caption{The simulated per stream SE for analog beam steering and HBF, together with the theoretical bounds for different values of $N_b$ given $N_t= 32$.}
    \label{fig_3}
\end{figure}

\begin{figure}[t!]
   \centering
   \includegraphics[scale=.49]{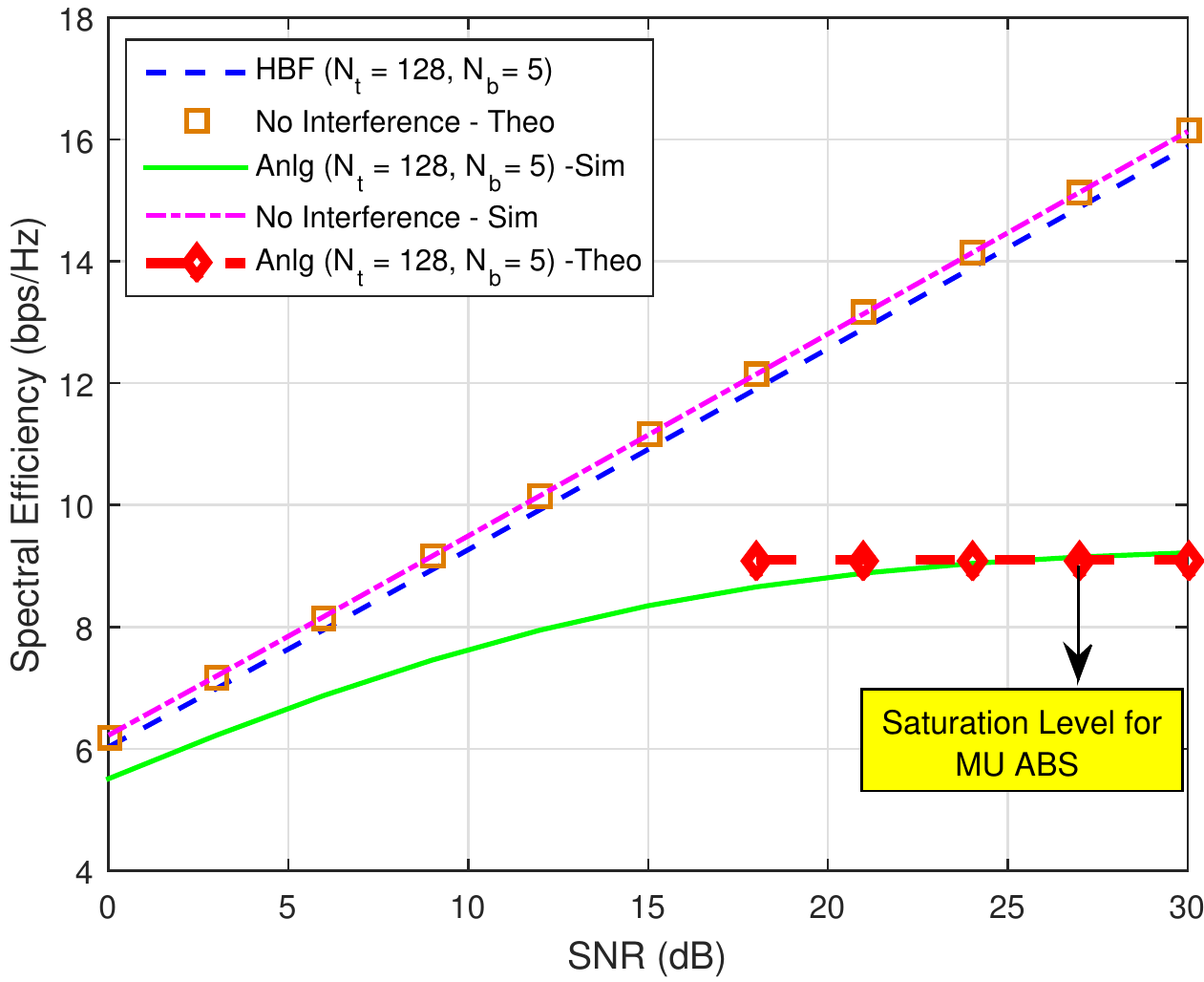}
    \caption{The simulated per stream SE for analog beam steering and HBF, together with the theoretical bounds for $N_b= 5$ given $N_t= 128$.}
    \label{fig_4}
\end{figure}
	
\section{Conclusion}
\label {conclusion}

In this paper, we provided a closed form approximation for the upper bound of the achievable per stream SE for MU ABS and proved that at high SNR regime, it converges to that upper bound and saturates. Also, we quantified the per stream SE gap between MU ABS and MU HBS mathematically at high SNR and showed that MU HBS scales with SNR and thus is favourable in high SNR regimes compared to MU ABS. Moreover, we validated with simulation results our proposed closed form approximations and defined the conditions for these approximations to be tight. 

\section{Acknowledgement}
\label{Acknowledgement}

This work has received a French state support granted to the CominLabs excellence laboratory and managed by the National Research Agency in the Investing for the Future program under reference Nb. ANR-10-LABX-07-01 and project name M\textsuperscript{5}HESTIA (\textbf{mm}Wave \textbf{M}ulti-user \textbf{M}assive \textbf{M}IMO \textbf{H}ybrid \textbf{E}quipments for \textbf{S}ounding, \textbf{T}ransmissions and HW \textbf{I}mplement\textbf{A}tion).

\bibliographystyle{IEEEtran}
\bibliography{main.bib}

\end{document}